\documentclass[a4paper, UKenglish, cleveref, autoref, thm-restate, pagebackref]{lipics-v2021}

\pdfoutput=1 
\hideLIPIcs  

\usepackage[table]{xcolor}
\definecolor[named]{colorSecondary}{rgb}{0.75, 0.35, 0.10}
\usepackage{booktabs}
\usepackage{multirow}
\usepackage{mathtools}
\usepackage{nicefrac}
\usepackage{xspace}
\usepackage{tikz}
\usetikzlibrary{patterns}
\usetikzlibrary{calc}
\usepackage{pgfplots}
\pgfplotsset{compat=1.18}
\usepackage{contour}
\usepackage[vlined, linesnumbered, ruled]{algorithm2e}
\SetEndCharOfAlgoLine{}
\SetCommentSty{textrm}
\SetArgSty{normalfont}
\SetKwProg{Fn}{Function}{}{}
\SetKw{Continue}{continue}

\usepackage{todonotes}

\renewcommand{\r}[1]{\mathopen{}\left(#1\right)\mathclose{}}
\newcommand{\s}[1]{\mathopen{}\left[#1\right]\mathclose{}}
\renewcommand{\c}[1]{\mathopen{}\left\{#1\right\}\mathclose{}}
\newcommand{\abs}[1]{\mathopen{}\left|#1\right|\mathclose{}}

\newcommand{\fu}[2]{#1\r{#2}}

\newcommand{\op}[2]{#1\c{#2}}
\renewcommand{\O}[1]{\fu{\mathcal{O}}{#1}}

\newcommand{\f}[1]{\fu{f}{#1}}
\newcommand{\sig}[1]{\fu{\sigma}{#1}}

\newcommand{\mi}[1]{\op{\min}{#1}}
\newcommand{\cov}[1]{\fu{C}{#1}}

\DeclareMathOperator{\tableopt}{opt}
\newcommand{\opt}[1]{%
  \if\relax\detokenize{#1}\relax
    \tableopt%
  \else
    \tableopt\mathopen{}\r{#1}\mathclose{}%
  \fi
}

\DeclareMathOperator{\tableOPT}{OPT}
\newcommand{\OPT}[1]{%
  \if\relax\detokenize{#1}\relax
    \tableOPT%
  \else
    \tableOPT\mathopen{}\r{#1}\mathclose{}%
  \fi
}

\newcommand{\C}{\mathcal{C}}
\newcommand{\E}{\mathcal{E}}
\newcommand{\I}{\mathcal{I}}

\newcommand{\real}{\mathbb{R}}
\newcommand{\interval}[1]{\left[#1\right)}
\newcommand{\nonneg}{\interval{0, \infty}}
\newcommand{\es}{\emptyset}

\newcommand{\problemtext}[1]{{\normalfont\textsc{#1}}}
\newcommand{\IO}[1]{{\normalfont\( #1 \)\nobreakdash-\problemtext{Interval Ordering}}\xspace}
\newcommand{\fIO}{\IO{f}}

\newcommand{\NPhard}[1]{{\normalfont NP\nobreakdash-hard#1}}

\newcommand{\citetdurr}{Dürr~et~al.~\cite{durr2012interval}}

\let\tfrac\nicefrac

\newcommand{\dashedline}[7]{
  \draw[
    line width=#1,
    preaction={draw=white, line width=#1}
  ] let
    \p1 = (#2),
    \p2 = (#3),
    \n1 = {\figscale * veclen(\x2 - \x1, \y2 - \y1)}, 
    \n2 = {#6*#4 + (1 - #7)*#5}, 
    \n3 = {round((\n1 + \n2) / (#4 + #5)) * (#4 + #5) - \n2}, 
    \n4 = {\n1 / \n3} 
  in
    [
      dash pattern=on {\n4*#4} off {\n4*#5},
      dash phase={#6*\n4*#4}
    ] (#2) -- (#3);
}

\newcounter{problem}
\crefname{problem}{Problem}{Problems}
\creflabelformat{problem}{#2\circled{#1}#3}
\makeatletter
\newcommand*\circled[1]{\raisebox{.1\height}{\scalebox{0.8}{\tikz[baseline=(char.base)]{\node[shape=circle, draw, inner sep=.75pt] (char) {#1};}}}\hspace{.5pt}}
\newcommand{\pinput}[1]{\unskip\normalsize\textcolor{white}{\ensuremath{\blacktriangleright}}\nobreakspace\textcolor{lipicsGray}{\sffamily\textbf{Input: }} & \normalsize #1 \smallskip \\}
\newcommand{\pquestion}[2][Question]{\normalsize\textcolor{white}{\ensuremath{\blacktriangleright}}\nobreakspace\textcolor{lipicsGray}{\sffamily\textbf{#1: }} & \normalsize#2}
\newenvironment{problem}[1][]
{%
  \center%
    \minipage{1\textwidth}
      \refstepcounter{problem}
      \phantomsection
      \protected@edef\@currentlabelname{\problemtext{#1}}
      \normalsize\textcolor{lipicsGray}{\ensuremath{\blacktriangleright}}\nobreakspace
      {\sffamily\bfseries Problem~\theproblem. }\ignorespaces
      \problemtext{#1}.\medskip\\
      \tabularx{\textwidth}{@{}l@{}X}
}{%
      \endtabularx
    \endminipage
  \endcenter
}
\makeatother

\title{Computational Complexity of the Interval Ordering Problem} 

\titlerunning{Computational Complexity of the Interval Ordering Problem} 

\author{Simeon Pawlowski}{Technische Universität Berlin, Faculty IV, Institute of Software Engineering and Theoretical Computer Science, Algorithmics and Computational Complexity, Germany}{s.pawlowski@campus.tu-berlin.de}{}{}

\author{Vincent Froese}{Technische Universität Berlin, Faculty IV, Institute of Software Engineering and Theoretical Computer Science, Algorithmics and Computational Complexity, Germany}{vincent.froese@tu-berlin.de}{}{}

\authorrunning{S. Pawlowski and V. Froese} 

\Copyright{Simeon Pawlowski and Vincent Froese} 

\ccsdesc[500]{Theory of computation~Dynamic programming}
\ccsdesc[100]{Theory of computation~Fixed parameter tractability}
\ccsdesc[500]{Theory of computation~Problems, reductions and completeness}
\ccsdesc[300]{Applied computing~Bioinformatics}

\keywords{dynamic programming, polynomial-time algorithms, NP-hardness, running time lower bounds, protein folding} 

\acknowledgements{We thank anonymous reviewers for their valuable feedback.}

\nolinenumbers 

\EventEditors{John Q. Open and Joan R. Access}
\EventNoEds{2}
\EventLongTitle{42nd Conference on Very Important Topics (CVIT 2016)}
\EventShortTitle{CVIT 2016}
\EventAcronym{CVIT}
\EventYear{2016}
\EventDate{December 24--27, 2016}
\EventLocation{Little Whinging, United Kingdom}
\EventLogo{}
\SeriesVolume{42}
\ArticleNo{23}

\begin{document}

  \maketitle

  \begin{abstract}
    We study an interval ordering problem introduced by Dürr et al.~[Discrete Appl.\ Math.~2012] which is motivated by applications in bioinformatics.
    The task is to order a given set of~\( n \) intervals with the goal of minimizing a certain objective which is defined via a given cost function~\( f \) which assigns a cost to the exposed part of each interval (that is, the pieces not covered by previous intervals).
    We develop a dynamic programming approach which solves the problem with~\( \O{2^n\cdot\fu{\text{poly}}{n}} \) oracle calls to~\( f \) and arithmetic operations.
    Moreover, our approach yields polynomial-time algorithms for all cost functions~\( f \) such that~\( f- \f{0} \) is subadditive or superadditive.
    This answers an open question for the function~\( \f{x} = 2^x \).
    We contrast these results by proving a running time lower bound of~\( 2^{n-1} \) for any algorithm that solves the problem for every function~\( f \) (with oracle access only) and further proving NP-hardness for some classes of simple functions.
    Thus, we significantly narrow the gap regarding the computational complexity of the problem.
  \end{abstract}

  \section{Introduction}
    \label{sec:introduction}

    We study the \fIO problem introduced by~\citetdurr{} where we are given a set~\( \I \) of~\( n \) intervals \( I_j = \interval{a_j, b_j} \subset \real \) with~\( a_j < b_j \) for \( j = 1, \ldots, n \) and the task is to find a total ordering of the intervals that minimizes the overall costs depending on the ``exposed parts'' of the intervals.
    For an ordering~\( \sigma = \r{\sig{1}, \ldots, \sig{n}} \in S_n \) (where~\( S_n \) denotes the set of permutations of~\( \c{1, \ldots, n} \)), the exposed part of an interval~\( I_j \) is~\( E_{\sigma, j} \coloneqq I_j\setminus\bigcup_{i:\fu{\sigma^{-1}}{i}<\fu{\sigma^{-1}}{j}}{I_i} \), that is, the subset of~\( I_j \) which is not ``covered'' by intervals coming before~\( I_j \) in~\( \sigma \) (see \cref{fig:fio-example} for an illustration).
    Note that~\( E_{\sigma, j} \) is a union of disjoint intervals.
    Its \emph{length}~\( \abs{E_{\sigma, j}} \) is defined as the total sum of the lengths of its individual intervals (an interval~\( \interval{a, b} \) has length~\( \abs{\interval{a, b}} \coloneqq b - a \)).
    Now, a given cost function~\( f\colon\nonneg \to \real \) assigns a cost to each interval depending on the length of its exposed part.
    Formally, the decision problem is defined as follows.

    \begin{problem}[{\ensuremath{f}{\normalfont-}\nolinebreak{}Interval Ordering}]
      \label{pr:f-interval-ordering}
      \pinput{A set \( \I = \c{I_1, \ldots, I_n} \) of non-empty intervals \( I_j = \interval{a_j, b_j} \subset \real \) and a number~\( W \in \real \).}
      \pquestion{Is there an ordering~\( \sigma \in S_n \) such that \( \sum_{j =1}^n\f{\abs{E_{\sigma, j}}} \le W \)?}
    \end{problem}\noindent

    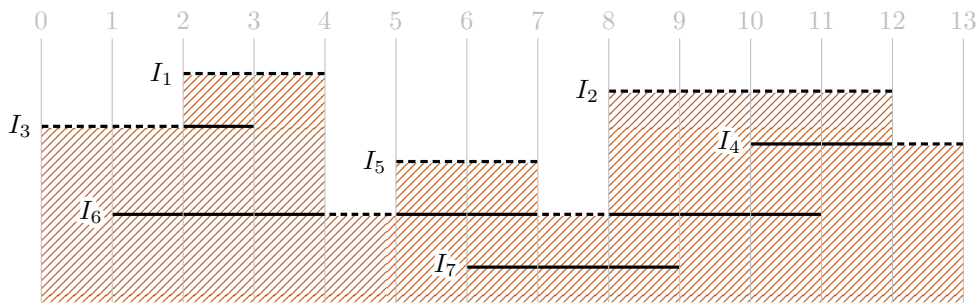
\begin{figure}[btp]
        \centering
        \pgfmathsetmacro{\figscale}{\textwidth/15cm}
          \begin{tikzpicture}[scale=\figscale]
            \draw[color=white] (-1, 0) -- (14, 0);

            \coordinate (a1)  at ( 2, 3.25);
            \coordinate (a13) at ( 2, 2.5 );
            \coordinate (b1)  at ( 4, 3.25);
            \coordinate (b16) at ( 4, 1.25);
            \coordinate (a2)  at ( 8, 3   );
            \coordinate (a26) at ( 8, 1.25);
            \coordinate (b2)  at (12, 3   );
            \coordinate (b24) at (12, 2.25);
            \coordinate (a3)  at ( 0, 2.5 );
            \coordinate (a30) at ( 0, 0   );
            \coordinate (b3)  at ( 3, 2.5 );
            \coordinate (a4)  at (10, 2.25);
            \coordinate (b4)  at (13, 2.25);
            \coordinate (b40) at (13, 0   );
            \coordinate (a5)  at ( 5, 2   );
            \coordinate (a56) at ( 5, 1.25);
            \coordinate (b5)  at ( 7, 2   );
            \coordinate (b56) at ( 7, 1.25);
            \coordinate (a6)  at ( 1, 1.25);
            \coordinate (b6)  at (11, 1.25);
            \coordinate (a7)  at ( 6, 0.5 );
            \coordinate (b7)  at ( 9, 0.5 );

            \fill[
                pattern=north east lines,
                pattern color=colorSecondary
            ]  (a30) -- (a3)%
                     -- (a13) -- (a1)%
                              -- (b1)%
                 -- (b16)%
                 -- (a56) -- (a5)%
                          -- (b5)%
                 -- (b56)%
                 -- (a26)     -- (a2)%
                              -- (b2)%
                    -- (b24)%
                    -- (b4)%
            -- (b40);

            \dashedline{1.2}      {a1}      {b1}{2.7}{1.8}{0}{0}
            \dashedline{1.2}      {a2}      {b2}{2.7}{1.8}{0}{0}
            \draw[line width=1.2] (a13) --  (b3);
            \dashedline{1.2}      {a3}     {a13}{2.7}{1.8}{0}{1}
            \draw[line width=1.2] (a4)  -- (b24);
            \dashedline{1.2}      {b24}    {b4}{2.7}{1.8}{1}{0}
            \dashedline{1.2}      {a5}     {b5}{2.7}{1.8}{0}{0}
            \draw[line width=1.2] (a6)  -- (b16);
            \dashedline{1.2}      {b16}    {a56}{2.7}{1.8}{1}{1}
            \draw[line width=1.2] (a56) -- (b56);
            \dashedline{1.2}      {b56}    {a26}{2.7}{1.8}{1}{1}
            \draw[line width=1.2] (a26) --  (b6);
            \draw[line width=1.2] (a7)  --  (b7);

            \contourlength{1.5pt}
            \contournumber{64}
            \coordinate[label=left:\( I_1 \)]                    (l1)  at  (a1);
            \coordinate[label=left:\( I_2 \)]                    (l2)  at  (a2);
            \coordinate[label=left:\( I_3 \)]                    (l3)  at  (a3);
            \coordinate[label=left:{\contour{white}{\( I_4 \)}}] (l4)  at  (a4);
            \coordinate[label=left:\( I_5 \)]                    (l5)  at  (a5);
            \coordinate[label=left:{\contour{white}{\( I_6 \)}}] (l6)  at  (a6);
            \coordinate[label=left:{\contour{white}{\( I_7 \)}}] (l7)  at  (a7);

            \draw[line width=0.4, color=gray!50] ( 0, 0) -- ( 0, 3.75);
            \draw[line width=0.4, color=gray!50] ( 1, 0) -- ( 1, 3.75);
            \draw[line width=0.4, color=gray!50] ( 2, 0) -- ( 2, 3.75);
            \draw[line width=0.4, color=gray!50] ( 3, 0) -- ( 3, 3.75);
            \draw[line width=0.4, color=gray!50] ( 4, 0) -- ( 4, 3.75);
            \draw[line width=0.4, color=gray!50] ( 5, 0) -- ( 5, 3.75);
            \draw[line width=0.4, color=gray!50] ( 6, 0) -- ( 6, 3.75);
            \draw[line width=0.4, color=gray!50] ( 7, 0) -- ( 7, 3.75);
            \draw[line width=0.4, color=gray!50] ( 8, 0) -- ( 8, 3.75);
            \draw[line width=0.4, color=gray!50] ( 9, 0) -- ( 9, 3.75);
            \draw[line width=0.4, color=gray!50] (10, 0) -- (10, 3.75);
            \draw[line width=0.4, color=gray!50] (11, 0) -- (11, 3.75);
            \draw[line width=0.4, color=gray!50] (12, 0) -- (12, 3.75);
            \draw[line width=0.4, color=gray!50] (13, 0) -- (13, 3.75);

            \coordinate[label=above:\color{gray!50}{\(  0 \)}]  (i0) at ( 0, 3.75);
            \coordinate[label=above:\color{gray!50}{\(  1 \)}]  (i1) at ( 1, 3.75);
            \coordinate[label=above:\color{gray!50}{\(  2 \)}]  (i2) at ( 2, 3.75);
            \coordinate[label=above:\color{gray!50}{\(  3 \)}]  (i3) at ( 3, 3.75);
            \coordinate[label=above:\color{gray!50}{\(  4 \)}]  (i4) at ( 4, 3.75);
            \coordinate[label=above:\color{gray!50}{\(  5 \)}]  (i5) at ( 5, 3.75);
            \coordinate[label=above:\color{gray!50}{\(  6 \)}]  (i6) at ( 6, 3.75);
            \coordinate[label=above:\color{gray!50}{\(  7 \)}]  (i7) at ( 7, 3.75);
            \coordinate[label=above:\color{gray!50}{\(  8 \)}]  (i8) at ( 8, 3.75);
            \coordinate[label=above:\color{gray!50}{\(  9 \)}]  (i9) at ( 9, 3.75);
            \coordinate[label=above:\color{gray!50}{\( 10 \)}] (i10) at (10, 3.75);
            \coordinate[label=above:\color{gray!50}{\( 11 \)}] (i11) at (11, 3.75);
            \coordinate[label=above:\color{gray!50}{\( 12 \)}] (i12) at (12, 3.75);
            \coordinate[label=above:\color{gray!50}{\( 13 \)}] (i13) at (13, 3.75);

          \end{tikzpicture}

      \caption{%
        A visualization of an interval ordering of the intervals \( I_1, \ldots, I_7 \) from top to bottom.
        The exposed parts of the intervals are drawn with dashed lines.
        For example, the interval \( I_6 = \interval{1, 11} \) has the exposed part \( I_6\setminus \bigcup_{i=1}^5 I_i = \interval{4, 5} \cup \interval{7, 8} \).
        The other parts of \( I_6 \) (solid lines) are not exposed as they are covered by the intervals \( I_1, \ldots, I_5 \) which come before \( I_6 \) in the ordering.
      }

      \label{fig:fio-example}
    \end{figure}

    The \fIO problem is motivated by a problem in biochemistry called protein folding, that is, computing the spatial structure of a protein (a large molecule consisting of many different atoms linked together).
    A crucial part of that problem is to reconstruct the protein backbone, that is, to determine the position of the main chain of atoms.
    By physically measuring distances between some pairs of atoms, one can reconstruct their spatial structure (a special case of the well-studied distance geometry problem~\cite{braun1987distance, duxbury2016unassigned, lavor2012recent, liberti2021cycle, more1997global}).
    This problem can be simplified to a certain bitstring reconstruction problem where the goal is to find a binary string such that for given index intervals, the corresponding substring equals an unknown string.
    Equality can only be tested via an oracle, which means that one has to try all possible substrings in the worst case.
    In short, \fIO with integral intervals and \( \f{x} = 2^x \) corresponds to finding a process order of the index intervals which minimizes the overall number of oracle tests (see Dürr et al.~\cite[Section~2]{durr2012interval} for a detailed exposition and further literature references).
    Besides this direct application, \fIO is a natural interval scheduling problem which might find further applications and which we believe to be of independent interest on its own.
    
    Dürr et al.~\cite{durr2012interval} showed that \fIO is \NPhard{} for a rather artificial (discontinuous non-monotonic) cost function~\( f \), but solvable in polynomial time if the input intervals are \emph{agreeable}, that is, there is an ordering induced by the left interval end points which is the same as an ordering induced by the right end points.
    They also prove polynomial-time solvability for \emph{laminar} intervals (that is, for every two intervals~\( I_i, I_j \), either \( I_i \cap I_j = \es \) or one is contained in the other) if the function~\( f - \f{0} \) is superadditive (note that their \NPhard{ness} result holds for laminar intervals).
    Most importantly, they posed the open question whether the problem is polynomial-time solvable for the cost function \( \f{x} = 2^x \), which constitutes the motivating use-case in biochemistry.

    \subparagraph*{Our Results.}
    In this work, we answer this question positively showing that \fIO can be solved in polynomial time for all functions~\( f \) for which~\( f - \f{0} \) is superadditive.
    Our main contribution is a dynamic programming algorithm described in \cref{sec:arbitrary-cost} that solves \fIO for arbitrary cost functions in \( \O{n^3 \cdot 2^n} \) time (assuming~\( f \) can be computed in constant time).
    The algorithm is based on the key idea to consider all possible exposed (and covered) parts instead of the intervals themselves.
    This view allows us to do a refined analysis of the structure of exposed parts in \cref{sec:polynomial} to obtain polynomial-time algorithms for all cost functions~\( f \) where~\( f - \f{0} \) is subadditive or superadditive.
    We also show that the case of pairwise intersecting intervals can be solved in polynomial time for arbitrary cost functions.

    We complement our algorithmic results in \cref{sec:hardness} by showing a worst-case running time lower bound of~\( 2^{n-1} \) for every algorithm solving \fIO for arbitrary~\( f \) with only oracle access to~\( f \).
    Notably, this result is also based on considering exposed parts.
    Hence, our results reveal that the exposed parts are essential for understanding the complexity of the problem.
    Moreover, we prove \NPhard{ness} for strictly monotone continuous piecewise linear functions and also for subadditive and superadditive functions.
    Thus, we draw a more precise picture of the complexity landscape for \fIO (see \cref{tab:results} for an overview).
    Finally, in \cref{sec:parameter}, we discuss some parameterized complexity results.

    \newcommand{\acc}{\rowcolor{colorSecondary!8}}
    \newcommand{\tblttl}[1]{{\selectfont\sffamily\bfseries#1}}
    \begin{table}[!t]
      \centering
      \footnotesize
      \def\arraystretch{1.3}
      \caption{%
        Overview of the computational complexity of \fIO (\( n \) denotes the number of input intervals).
        The running time to compute~\( f \) is neglected.
      }
      \belowrulesep = 5pt
      \setlength{\tabcolsep}{2.5mm}
      \begin{tabular}{l@{\hspace{5.5mm}}l@{\hspace{5.5mm}}l}
        \tblttl{Cost Function~\( f \)}     & \tblttl{Intervals}    & \tblttl{Complexity}                                             \vspace{2pt}      \\
        \midrule
        \acc \( f - \f{0} \) superadditive & laminar               & \makebox[1.5cm][l]{\( \O{n \log{n}} \)}   (\cite[Thm.~3]{durr2012interval})       \\
             continuous convex             & agreeable             & \makebox[1.5cm][l]{\( \O{n^2} \)}         (\cite[Thm.~2]{durr2012interval})       \\
        \acc arbitrary                     & agreeable             & \makebox[1.5cm][l]{\( \O{n^3} \)}         (\cite[Thm.~1]{durr2012interval})       \\
             \( f - \f{0} \) subadditive   & arbitrary             & \makebox[1.5cm][l]{\( \O{n^5} \)}         (\cref{cor:subadditive})                \\
        \acc \( f - \f{0} \) superadditive & arbitrary             & \makebox[1.5cm][l]{\( \O{n^5} \)}         (\cref{thm:superadditive})              \\
             arbitrary                     & pairwise intersecting & \makebox[1.5cm][l]{\( \O{n^6} \)}         (\cref{thm:pairwise-intersecting})      \\
        \acc strict.\ mon.\ cont.\ piecewise linear & laminar      & \makebox[1.5cm][l]{\NPhard{}}             (\cref{cor:general-reduction})          \\
             subadditive/superadditive     & laminar               & \makebox[1.5cm][l]{\NPhard{}}             (\cref{cor:sub-superadditive-hardness}) \\
        \acc arbitrary                     & laminar (integral)    & \makebox[1.5cm][l]{\NPhard{}}             (\cite[Thm.~8]{durr2012interval})       \\
             arbitrary                     & arbitrary             & \makebox[1.5cm][l]{\( \O{n^3\cdot2^n} \)} (\cref{thm:general-time})
      \end{tabular}
      \vspace{-1em}
      \label{tab:results}
    \end{table}

  \section{Preliminaries}
    \label{sec:preliminaries}

    For \( n \in \mathbb N \), we define~\( [n] \coloneqq \c{1, \ldots, n} \).
    A function~\( f\colon\nonneg\to\real \) is \emph{superadditive} (\emph{subadditive}) if~\( \f{x + y} \ge \f{x} + \f{y} \) (\( \f{x + y} \le \f{x} + \f{y} \)) for all~\( x, y \in \nonneg \).

    We consider intervals of the form \( I = \interval{a, b} \subset \real \) with \( a < b \in \real \).
    We call~\( a \) the \emph{start point} and~\( b \) the \emph{end point} of~\( I \).
    The \emph{length}~\( \abs{I} \) of~\( I \) is defined as \( b - a \).
    We say that two intervals~\( I_i \) and~\( I_j \) \emph{touch} each other if they are disjoint and their union is an interval.
    For a set \( \I = \c{I_1, \ldots, I_n} \) of intervals and a permutation~\( \sigma \in S_n \), we call the tuple \( \r{I_{\sig{1}}, \ldots, I_{\sig{n}}} \) a \emph{sequence} of~\( \I \).
    We say that \( I_i \) \emph{comes before} \( I_j \) if \( \fu{\sigma^{-1}}{i} < \fu{\sigma^{-1}}{j} \).
    For a subset~\( \I' \subseteq \I \), we call \( \cov{\I'} \coloneqq \bigcup_{I \in \I'}{I} \) the \emph{covered area} of~\( \I' \).
    Note that~\( \cov{\I'} \) is a union of some pairwise disjoint intervals.
    We call the intervals in the smallest set of intervals whose union equals~\( \cov{\I'} \) the \emph{interval components} of~\( \cov{\I'} \) (an empty covered area has no interval components).
    If~\( \cov{\I'} \) is an interval, then we call it a \emph{covered interval}.
    We define~\( \C_\I \coloneqq \c{\cov{\I'} \mid \I' \subseteq \I} \) and we denote by \( \C^*_\I \subseteq \C_\I \) the subset of all covered intervals.
    Note that the exposed part~\( E_{\sigma, j} \) of~\( I_j \) is exactly \( I_j\setminus \cov{\I'} \), where~\( \I' = \c{I_i \mid \fu{\sigma^{-1}}{i} < \fu{\sigma^{-1}}{j}} \) (see \cref{fig:fio-example}).
    Again, we call the intervals of the smallest set of intervals whose union is~\( E_{\sigma, j} \), the \emph{interval components} of~\( E_{\sigma, j} \).
    We define~\( \E_\I \coloneqq \c{E_{\sigma, j} \mid I_j \in \I, \sigma \in S_n}\setminus\c{\es} \).

    Throughout this work, we assume that the cost function~\( f\colon \nonneg \to \real \) is computable and we neglect the time required to compute it and also the running times for arithmetic operations.

  \section{Arbitrary Cost Functions}
    \label{sec:arbitrary-cost}

    In this section, we develop an exponential-time algorithm computing the optimal cost for arbitrary cost functions~\( f \).
    To start with, note that there is a simple standard dynamic programming approach to solve the problem in~\( \O{2^n\cdot\fu{\text{poly}}{n}} \) time:
    Use a table to store optimal costs for each subset~\( \I' \subseteq \I \) of input intervals.
    Clearly, each entry can be computed via
    \[
      \opt{\I'} = \min_{I_j \in \I'}\c{\opt{\I'\setminus\c{I_j}} + f\r{\abs{I_j\setminus\bigcup_{I_i in \I'\setminus\c{I_j}}I_i}}}\text{.}
    \]
    However, it is not clear how to adapt this approach to obtain polynomial time for special cost functions (since the table has exponential size).
    Therefore, we develop a different approach based on analyzing the structure of optimal solutions in terms of exposed and covered parts.
    This leads to a dynamic program with a polynomial-size table which allows us to derive polynomial-time algorithms for special cases in \Cref{sec:polynomial}.

    First, note that if~\( \cov{\I} \) has multiple interval components, then we can clearly solve the corresponding subinstances independently and sum up their optimal costs.
    Hence, we assume that~\( \cov{\I} \) is an interval.
    For \( X \subseteq\real \), let \( \I_X \coloneqq \c{I \in \I \mid I \subseteq X} \) be the subset of~\( \I \) induced by~\( X \) and let~\( \opt{\I'} \) denote the optimal cost of \( \I' \subseteq \I \) (where~\( \opt{\es} \coloneqq 0 \)).
    The following lemma describes how an optimal solution can be recursively computed based on the possible exposed parts~\( \E_\I \).

    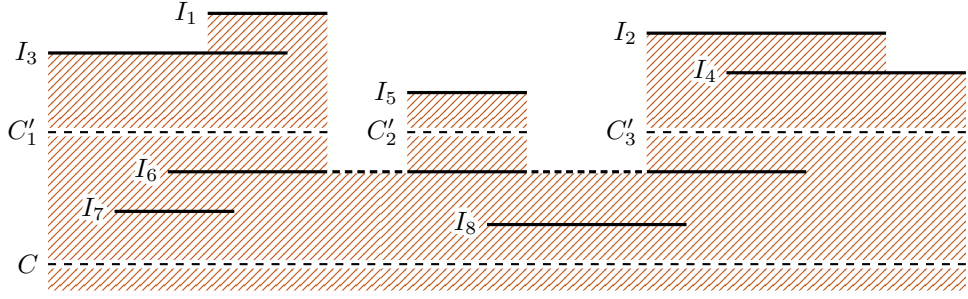
\begin{figure}
      \centering
      \pgfmathsetmacro{\figscale}{\textwidth/40cm}
        \begin{tikzpicture}[scale=\figscale]
          \draw[color=white] (-1, 0) -- (35.5, 0);

          \coordinate (a1)  at ( 6,   11.5 );
          \coordinate (a13) at ( 6,   10   );
          \coordinate (b1)  at (10.5, 11.5 );
          \coordinate (b17) at (10.5,  4   );
          \coordinate (a2)  at (22.5, 10.75);
          \coordinate (a27) at (22.5,  4   );
          \coordinate (b2)  at (31.5, 10.75);
          \coordinate (b24) at (31.5,  9.25);
          \coordinate (a3)  at ( 0,   10   );
          \coordinate (a30) at ( 0,    0   );
          \coordinate (b3)  at ( 9,   10   );
          \coordinate (a4)  at (25.5,  9.25);
          \coordinate (b4)  at (34.5,  9.25);
          \coordinate (b40) at (34.5,  0   );
          \coordinate (a5)  at (13.5,  8.5 );
          \coordinate (a57) at (13.5,  4   );
          \coordinate (b5)  at (18,    8.5 );
          \coordinate (b57) at (18.0,  4   );
          \coordinate (a6)  at ( 2.5,  5.5 );
          \coordinate (b6)  at ( 7,    5.5 );
          \coordinate (a7)  at ( 4.5,  4   );
          \coordinate (b7)  at (28.5,  4   );
          \coordinate (a8)  at (16.5,  2.5 );
          \coordinate (b8)  at (24,    2.5 );
          \coordinate (ca)  at ( 0,    1   );
          \coordinate (cb)  at (34.5,  1   );
          \coordinate (c1a) at ( 0,    7   );
          \coordinate (c1b) at (10.5,  7   );
          \coordinate (c2a) at (13.5,  7   );
          \coordinate (c2b) at (18,    7   );
          \coordinate (c3a) at (22.5,  7   );
          \coordinate (c3b) at (34.5,  7   );

          \fill[
              pattern=north east lines,
              pattern color=colorSecondary
          ]  (a30) -- (a3)%
                    -- (a13) -- (a1)%
                            -- (b1)%
                -- (b17)%
                -- (a57) -- (a5)%
                        -- (b5)%
                -- (b57)%
                -- (a27)     -- (a2)%
                            -- (b2)%
                  -- (b24)%
                  -- (b4)%
          -- (b40);

          \draw[line width=1.2] (a1)  --  (b1);
          \draw[line width=1.2] (a2)  --  (b2);
          \draw[line width=1.2] (a3)  --  (b3);
          \draw[line width=1.2] (a4)  --  (b4);
          \draw[line width=1.2] (a5)  --  (b5);
          \draw[line width=1.2] (a7)  -- (b17);
          \draw[line width=1.2] (a57) -- (b57);
          \draw[line width=1.2] (a27) --  (b7);
          \draw[line width=1.2] (a6)  --  (b6);
          \draw[line width=1.2] (a8)  --  (b8);

          \dashedline{1.2}{b17}{a57}{2.7}{1.8}{1}{1}
          \dashedline{1.2}{b57}{a27}{2.7}{1.8}{1}{1}

          \draw[line width=3.2, color=white] (ca)  --  (cb);
          \draw[line width=3.2, color=white] (c1a) -- (c1b);
          \draw[line width=3.2, color=white] (c2a) -- (c2b);
          \draw[line width=3.2, color=white] (c3a) -- (c3b);

          \dashedline{0.8}{ca}{cb}{3}{3}{0}{0}
          \dashedline{0.8}{c1a}{c1b}{3}{3}{0}{0}
          \dashedline{0.8}{c2a}{c2b}{3}{3}{0}{0}
          \dashedline{0.8}{c3a}{c3b}{3}{3}{0}{0}

          \contourlength{1.5pt}
          \contournumber{64}
          \coordinate[label=left:\( I_1 \)]                    (l1)  at  (a1);
          \coordinate[label=left:\( I_2 \)]                    (l2)  at  (a2);
          \coordinate[label=left:\( I_3 \)]                    (l3)  at  (a3);
          \coordinate[label=left:{\contour{white}{\( I_4 \)}}] (l4)  at  (a4);
          \coordinate[label=left:\( I_5 \)]                    (l5)  at  (a5);
          \coordinate[label=left:{\contour{white}{\( I_6 \)}}] (l6)  at  (a6);
          \coordinate[label=left:{\contour{white}{\( I_7 \)}}] (l7)  at  (a7);
          \coordinate[label=left:{\contour{white}{\( I_8 \)}}] (l8)  at  (a8);
          \coordinate[label=left:\( C \)]                      (lc)  at  (ca);
          \coordinate[label=left:\( C'_1 \)]                   (lc1) at (c1a);
          \coordinate[label=left:\( C'_2 \)]                   (lc2) at (c2a);
          \coordinate[label=left:\( C'_3 \)]                   (lc3) at (c3a);

        \end{tikzpicture}

      \caption{%
        An example sequence \( \sigma = \r{I_1, \ldots, I_8} \) from top to bottom.
        The covered areas are shaded (intervals \( I_1, \ldots, I_8 \) cover the interval~\( C \)).
        The last interval in the sequence with non-empty exposed part is~\( I_7 \).
        The exposed part \( E_{\sigma, 7} \) of \( I_7 \) is dashed.
        The intervals before~\( I_7 \) cover the area \( C' = C_1' \cup C_2' \cup C_3' \).
        Note that~\( \I_{C'} \) contains all intervals except \( I_7 \) and \( I_8 \).
        The cost of this sequence is the sum of the costs of the subsequences that cover \( C_1', C_2', C_3' \), the cost of \( E_{\sigma, 7} \), and the cost of all completely covered intervals after~\( I_7 \) which are not covered by~\( C' \), that is, \( I_8 \).
      }
      \label{fig:decomposition}
    \end{figure}

    \begin{lemma}
      \label{lem:recursion}
      Let~\( \I \) be intervals such that~\( C \coloneqq \cov{\I} \) is an interval.
      Then,
      \[
        \opt{\I} = \min\c{\f{\abs{E}} + \r{\abs{\I} - \abs{\I_{C \setminus E}} - 1} \cdot \f{0} + \opt{\I_{C \setminus E}} \mid E \in \E_\I, C\setminus E \in \C_\I}\text{.}
      \]
    \end{lemma}
    \begin{proof}
      Consider an optimal ordering~\( \sigma \) and let the corresponding sequence be \( \r{I_1, \ldots, I_n} \) without loss of generality.
      Let~\( I_k \) be the last interval that is not completely covered, that is, \( E_{\sigma, k} \ne \es \) (see \Cref{fig:decomposition}).
      Further, let~\( \I' \coloneqq \c{I_1, \ldots, I_{k-1}} \) and \( C' \coloneqq \cov{\I'} = C \setminus E_{\sigma, k} \).
      Note that we can assume that no interval from \( I_{k+1}, \ldots, I_n \) is completely covered by~\( C' \) (that is, \( \I' = \I_{C'} \)) since otherwise we could move it directly before~\( I_k \) in~\( \sigma \) without changing the cost (since the cost is~\( \f{0} \) in both cases).

      Hence, we can decompose \( \opt{\I} \) as follows:
      First, we have the cost of the intervals \( I_{k+1}, \ldots, I_n \), which are all completely covered.
      Hence, their cost is \( \r{n - k}\cdot \f{0} = \r{\abs{\I} - \abs{\I_{C'}} - 1}\f{0} \).
      Second, we have the cost of~\( I_k \), which is \( \f{\abs{E_{\sigma, k}}} \).
      Finally, we have the cost of the subsequence \( \r{I_1, \ldots, I_{k-1}} \).
      Together, this yields
      \[
        \opt{\I} = \f{\abs{E_{\sigma, k}}} + \r{\abs{\I} - \abs{\I_{C'}} - 1}\cdot\f{0} + \opt{\I_{C'}}\text{,}
      \]
      which implies
      \[
        \opt{\I} \ge \mi{\f{\abs{E}} + \r{\abs{\I} - \abs{\I_{C \setminus E}} - 1} \cdot \f{0} + \opt{\I_{C \setminus E}} \mid E \in \E_\I, C\setminus E \in \C_\I}\text{.}
      \]
      The reverse inequality ``\( \le \)'' easily holds since for every~\( E \in \E_\I \) with~\( C\setminus E \in \C_\I \) there exists a sequence \( \sigma \) which achieves the cost
      \[
        \f{\abs{E}} + \r{\abs{\I} - \abs{\I_{C \setminus E}} - 1} \cdot \f{0} + \opt{\I_{C \setminus E}}\text{.}
      \]
      First, order all intervals in \( \I_{C\setminus E} \) optimally.
      Then, take a remaining interval which covers~\( E \) followed by all other remaining intervals.
    \end{proof}

    Note that~\( \opt{\I_{C \setminus E}} \) can again be computed as the sum of optimal costs of the individual interval components of \( C \setminus E \).
    That is, to compute \( \opt{\I_C} \) for some covered interval~\( C \), we only need to compute optimal costs of subinstances corresponding to proper subintervals of~\( C \).
    This leads to a dynamic program which computes optimal costs for all covered intervals in a topological order with respect to the proper subset relation.
    Let~\( \C^*_\I \subseteq \C_\I \) be the subset of all covered intervals and let the table \( \OPT{} \) contain an entry~\( \OPT{C} \coloneqq \opt{\I_C} \) for each~\( C \in \C^*_\I \).
    \cref{alg:compute-opt} shows the pseudocode for filling the table \( \OPT{} \).

    \begin{algorithm}[t]
      \caption{Computing optimal costs for arbitrary cost functions~\( f \)}
      \ForEach{covered interval \( C \in \C^*_\I \) in topological order}{%
        \( \OPT{C} \leftarrow \infty \)\;
        \( k \leftarrow \abs{\I_C} - 1 \)\;
        \ForEach{\( E \in \E_{\I} \)\label{line:min-loop}}{%
          \If{there is no \( I \in \I \) with \( E \subseteq I \subseteq C \)\label{line:checkE}}{%
            \Continue with next~\( E \)\;
          }
          \( C' \leftarrow C \setminus E \)\;\label{line:c-e}
          \( \omega \leftarrow \f{\abs{E}} \)\;
          \ForEach{interval component~\( C_l' \) of~\( C' \)}{%
            \If{\( C'_l \notin \C^*_\I \)\label{line:checkC'}}{%
                \Continue with next~\( E \)\;
              }
              \( k \leftarrow k - \abs{\I_{C'_l}} \)\;
              \( \omega \leftarrow \omega + \OPT{C'_l} \)\;\label{line:lookup}
            }
          \( \omega \leftarrow \omega + k \cdot \f{0} \)\;
          \( \OPT{C} \leftarrow \mi{\OPT{C}, \omega} \)\;\label{line:min}
        }
      }
      \label{alg:compute-opt}
    \end{algorithm}

    \begin{lemma}
      \label{lem:DP-correct}
      \cref{alg:compute-opt} correctly computes the optimal costs.
    \end{lemma}
    \begin{proof}
      The table~\( \OPT{} \) is filled in a topological order on the covered intervals~\( \C^*_\I \).
      By \cref{lem:recursion}, for each~\( C \in \C^*_\I \), it holds:
      \begin{equation*}
        \OPT{C} = \min\c{\f{\abs{E}} + \r{\abs{\I_C} - \sum_{l = 1}^{m}{\abs{\I_{C'_l}}} - 1} \cdot \f{0} + \sum_{l = 1}^{m}{\OPT{C_l'}}}\text{,}
      \end{equation*}
      where~\( C'_1, \ldots, C'_m \subseteq C\setminus E \) are the~\( m \le n \) interval components of~\( C\setminus E \) and the minimum is considered over all exposed parts \( E \in \E_{\I_C} \) such that \( C\setminus E \in \C_{\I_C} \).
      This minimum is computed via the for-loop in Line~\ref{line:min-loop} which iterates over all exposed parts~\( E \in \E_\I \).
      In order to check that~\( E \in \E_{\I_C} \), we first check in Line~\ref{line:checkE} whether~\( E \subseteq I \subseteq C \) for some~\( I \in \I \) and the condition~\( C\setminus E \in \C_{\I_C} \) is checked in Line~\ref{line:checkC'} by testing each interval component for containment in~\( \C^*_\I \).
      Note that in the case that~\( E = C \) we set~\( C' \) to~\( \es \) in Line~\ref{line:c-e} and~\( C' \) has no interval components.
      Thus, we correctly use the cost~\( \omega = \f{\abs{C}} + \r{\abs{\I_C} - 1}\cdot\f{0} \) for the minimum in Line~\ref{line:min} updating~\( \OPT{C} \).
      The values~\( \OPT{C_l'} \) are already computed when accessed in Line~\ref{line:lookup}
      since we iterate over~\( \C^*_\I \) in topological order.
      Hence, after termination, the correct optimal cost of~\( \I \) is stored in~\( \OPT{\cov{\I}} \) (assuming that \( \cov{\I} \) is an interval).
    \end{proof}

    Note that an optimal interval ordering can also be reconstructed when storing the optimal exposed part~\( E \) for each covered interval~\( C \) in \cref{alg:compute-opt}.

    As regards the running time, observe first that \( \abs{\C^*_\I}\in\O{n^2} \) since each covered interval starts at one of the~\( n \) start points and ends at one of the~\( n \) end points of the intervals in~\( \I \).
    We run a simple preprocessing to compute the set~\( \C^*_\I \).
    \begin{lemma}
      \label{lem:preproc}
      We can compute all covered intervals \( C\in\C^*_\I \) (stored in topological order with respect to proper subset relation) and the sizes~\( \abs{\I_C} \) in~\( \O{n^3} \) time.
    \end{lemma}
    \begin{proof}
      We first sort the input intervals by start point and by end point in \( \O{n\log n} \) time.
      For each start point~\( a \) and each endpoint~\( b > a \), we check whether \( \interval{a, b} \) is a covered interval in \( \O{n} \) time as follows:
      Clearly, we can ignore all intervals starting before~\( a \) or ending after~\( b \).
      For the remaining intervals, we need to check whether they cover~\( \interval{a, b} \).
      To this end, we iterate over all start and end points in ascending order and keep track of how many intervals have started but have not ended yet.
      If at some point~\( p \) between~\( a \) and~\( b \) all intervals have ended and no interval is starting at~\( p \), then the interval~\( \interval{a, b} \) is not covered.
      Otherwise, we store the covered interval~\( \interval{a, b} \) and also the number~\( \abs{\I_{\interval{a, b}}} \) of contained intervals.

      For the topological sorting, we sort all covered intervals in~\( \C^*_\I \) in \( \O{n^2\log n} \) time such that start points decrease and end points increase for identical start points.
    \end{proof}

    The crucial part of the running time turns out to be the computation of the set~\( \E_\I \) of exposed parts.

    \begin{lemma}
      \label{lem:compute-exposed}
      The set~\( \E_\I \) has size \( \O{2^n} \) and can be computed in \( \O{n^2 \cdot 2^n} \) time.
    \end{lemma}
    \begin{proof}
      Assume that the intervals \( \I = \c{I_1, \ldots, I_n} \) are indexed in descending order of their lengths.
      We build the set~\( \E_\I \) iteratively.
      Let \( \I_k \coloneqq \c{I_1, \ldots, I_k} \) for all \( k \in \s{n} \) and consider the set~\( \E_{\I_k} = \c{E_{\sigma, j} \mid I_j \in \I_k, \sigma \in S_k} \).
      Clearly, \( \mathcal E_{\I_1} = \c{I_1} \).
      We compute the set \( \mathcal E_{\I_{k+1}} \) from \( \mathcal E_{\I_k} \).
      To this end, note that since \( \abs{I_i} \geq \abs{I_{k+1}} \) for all \( i \in \s{k} \), there are no proper subintervals of~\( I_{k+1} \) in~\( \I_k \).
      Therefore, the exposed part of~\( I_{k+1} \) in any sequence of~\( \I_{k+1} \) is either empty or an interval
      \begin{itemize}
        \item whose start point is either the start point of~\( I_{k+1} \) or an end point of some interval in~\( \I_k \)
        \item and whose end point is either the endpoint of~\( I_{k+1} \) or a start point of an interval in~\( \I_k \).
      \end{itemize}
      Thus, there are at most \( \r{k+1}^2 \) non-empty exposed parts of~\( I_{k+1} \) in~\( \E_{\I_{k+1}} \).
      Moreover, for every \( E \in\E_{\I_k} \), there are the two (possibly identical or empty) exposed parts~\( E \) and \( E \setminus I_{k+1} \) in~\( \E_{\I_{k+1}} \).
      
      For the size of \( \E_{\I_{k+1}} \), we thus obtain the recursive relation \( \abs{\E_{\I_{k+1}}} \leq 2 \cdot \abs{\E_{\I_k}} + \r{k+1}^2 \) for all \( k \in \s{n - 1} \), where \( \abs{\E_{\I_1}} = 1 \).
      We now show inductively that \( \abs{\E_{\I_k}} \leq 6 \cdot 2^k - k^2 - 4k - 6 \) for all \( k \in \c{1, \ldots, n} \).
      For~\( k = 1 \), this is clear.
      For \( k \ge 2 \), we have:
      \begin{align*}
        \label{eq:recurrence-relation}
        \abs{\E_{\I_k}} & \leq 2 \cdot \abs{\E_{\I_{k-1}}} + k^2                            \\
                        & \leq 2 \cdot \r{6 \cdot 2^{k-1} - \r{k-1}^2 - 4\r{k-1} - 6} + k^2 \\
                        & = 6 \cdot 2^k - 2k^2 + 4k - 2 - 8k + 8 - 12 + k^2                 \\
                        & = 6 \cdot 2^k - k^2 - 4k - 6\text{.}
      \end{align*}
      Hence, \( \abs{\E_\I} = \abs{\E_{\I_n}} \in \O{2^n} \).

      We iteratively compute the sets~\( \E_{\I_1}, \ldots, \E_{\I_n} \) in~\( n \) rounds as described above.
      In round~\( k \ge 2 \), we compute \( \E_{\I_k} \) as follows: For each \( E \in \E_{\I_{k-1}} \), we compute the exposed part \( E \setminus I_{k} \) in \( \O{k} \) time (since~\( E \) has at most~\( k \) interval components).
      Additionally, we compute the exposed parts of~\( I_{k} \), by checking in constant time for each of the \( \O{k^2} \) pairs of start and end points from~\( \I_{k} \) whether they form an exposed part of \( I_{k} \), that is, a subinterval of~\( I_{k} \).
      To prevent duplicates, we can keep the exposed parts sorted in a balanced binary search tree, which allows to detect duplicates and to insert exposed parts in \( \O{n\cdot\log\abs{\E_\I}} \) time, that is, \( \O{n^2} \) time (the factor~\( \O{n} \) accounts for the time to compare two exposed parts).
      Overall, we compute at most \( 2^n \) exposed parts and each exposed part takes at most \( \O{n^2} \) time.
    \end{proof}

    With \cref{lem:compute-exposed}, we can state the time complexity of \cref{alg:compute-opt}.

    \begin{theorem}
      \label{thm:general-time}
      \cref{alg:compute-opt} solves \fIO{} in \( \O{n^3 \cdot 2^n} \) time.
    \end{theorem}
    \begin{proof}
      The preprocessing to compute the covered intervals~\( \C^*_\I \) (and the corresponding sizes~\( \abs{\I_C} \)) can be done in~\( \O{n^3} \) time (\Cref{lem:preproc}) and the exposed parts~\( \E_\I \) can be precomputed in~\( \O{n^2\cdot 2^n} \) time (\cref{lem:compute-exposed}).
      The table~\( \OPT{} \) has~\( \O{n^2} \) entries and each entry can be computed in~\( \O{n\cdot 2^n} \) time since~\( \abs{\E_\I} \in \O{2^n} \) and we need~\( \O{n} \) time for the check in Line~\ref{line:checkE}, computing~\( C\setminus E \) and~\( \abs{E} \), and iterating over the interval components of~\( C\setminus E \).
    \end{proof}

    We have seen that the exponential running time of \cref{alg:compute-opt} is due to the number of possible exposed parts.
    In the next section, we show that for certain cost functions we do not need to consider all possible exposed parts but only a polynomial number.
    In \cref{sec:hardness}, however, we show that for arbitrary~\( f \) (with only oracle access) it cannot be avoided to enumerate all possible exposed parts.

  \section{Polynomial-Time Solvable Cases}
    \label{sec:polynomial}

    In this section, we use the approach described in \cref{sec:arbitrary-cost} to obtain polynomial-time algorithms for special cost functions answering the open question by~\citetdurr{} regarding the function \( \f{x} = 2^x \).
    In particular, we give an \( \O{n^5} \)-time algorithm for every cost function~\( f \) where~\( f - \f{0} \) is superadditive.

    \subsection{Cost Functions where~\( f - \f{0} \) is Subadditive}
      \label{ssec:subadditive}

      First, we have a look at cost functions~\( f \) where~\( f - \f{0} \) is subadditive, that is, \( \f{a} + \f{b} \geq \f{a + b} + \f{0} \) holds for all \( a, b \in \nonneg \) (for example, every concave function).
      For such functions, we show the following.

      \begin{lemma}
        \label{lem:subadd-property}
        Let~\( f \) be a cost function such that~\( f - \f{0} \) is subadditive.
        Then there is an optimal sequence where all exposed parts are intervals.
      \end{lemma}
      \begin{proof}
        Assume that \( \r{I_1, \ldots, I_n} \) is an optimal sequence where \( I_j \) has an exposed part~\( E_j \) that is not an interval, that is, it has at least two interval components.
        This means that the covered area of \( \c{I_1, \ldots, I_{j-1}} \) has an interval component that is a subinterval of~\( I_j \).
        Therefore, the intervals before~\( I_j \) that cover this interval component are also subintervals of~\( I_j \) (see \cref{fig:subadditive}).
        Let those intervals be \( I_{i_1}, \ldots, I_{i_k} \) and let their exposed parts be \( E_{i_1}, \ldots, E_{i_k} \).
        Now consider what happens if we change \( I_{i_1}, \ldots, I_{i_k} \) to come after~\( I_j \).
        Then the exposed parts of \( I_{i_1}, \ldots, I_{i_k} \) will be empty and the exposed part of~\( I_j \) will increase by~\( \abs{E_{i_1}} + \cdots + \abs{E_{i_k}} \).
        Subadditivity now yields
        \[
          \f{\abs{E_{i_1}}} + \cdots + \f{\abs{E_{i_k}}} + \f{\abs{E_j}} \geq \f{\abs{E_{i_1}} + \cdots + \abs{E_{i_k}} + \abs{E_j}} + k \cdot \f{0}\text{,}
        \]
        that is, the cost does not increase.
        This rearrangement reduces the number of interval components of the exposed part of~\( I_j \) while not increasing the number of interval components of any other exposed part.
        Therefore, after a finite number of rearrangements we obtain an optimal sequence where all exposed parts are intervals.
      \end{proof}

      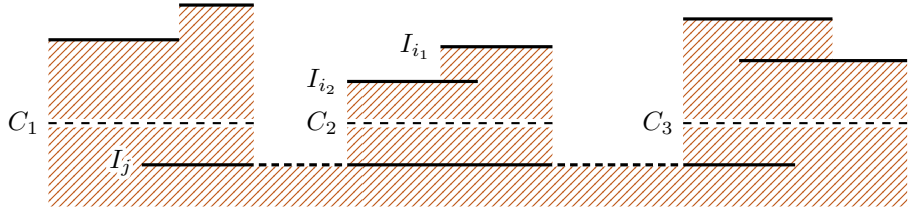
\begin{figure}[btp]
        \centering
        \pgfmathsetmacro{\figscale}{\textwidth/57cm}
          \begin{tikzpicture}[scale=\figscale]
            \draw[color=white] (0, 0) -- (52, 0);

            \coordinate (a1)  at (10, 10.875);
            \coordinate (a13) at (10,  9    );
            \coordinate (b1)  at (14, 10.875);
            \coordinate (b17) at (14,  2.25 );
            \coordinate (a2)  at (37, 10.125);
            \coordinate (a27) at (37,  2.25 );
            \coordinate (b2)  at (45, 10.125);
            \coordinate (b24) at (45,  7.875);
            \coordinate (a3)  at ( 3,  9    );
            \coordinate (a30) at ( 3,  0    );
            \coordinate (b3)  at (10,  9    );
            \coordinate (a4)  at (40,  7.875);
            \coordinate (b4)  at (49,  7.875);
            \coordinate (b40) at (49,  0    );
            \coordinate (a5)  at (19,  6.75 );
            \coordinate (a57) at (19,  2.25 );
            \coordinate (b5)  at (26,  6.75 );
            \coordinate (a6)  at (24,  8.625);
            \coordinate (a65) at (24,  6.75 );
            \coordinate (b6)  at (30,  8.625);
            \coordinate (b67) at (30,  2.25 );
            \coordinate (a7)  at ( 8,  2.25 );
            \coordinate (b7)  at (43,  2.25 );
            \coordinate (c1a) at ( 3,  4.5  );
            \coordinate (c1b) at (14,  4.5  );
            \coordinate (c2a) at (19,  4.5  );
            \coordinate (c2b) at (30,  4.5  );
            \coordinate (c3a) at (37,  4.5  );
            \coordinate (c3b) at (49,  4.5  );

            \fill[
                pattern=north east lines,
                pattern color=colorSecondary
            ]  (a30) -- (a3)%
                     -- (a13) -- (a1)%
                              -- (b1)%
                 -- (b17)%
                 -- (a57) -- (a5)%
                          -- (a65) -- (a6)%
                                   -- (b6)%
                 -- (b67)%
                 -- (a27)     -- (a2)%
                              -- (b2)%
                    -- (b24)%
                    -- (b4)%
            -- (b40);

            \draw[line width=1.2] (a1)  --  (b1);
            \draw[line width=1.2] (a2)  --  (b2);
            \draw[line width=1.2] (a3)  --  (b3);
            \draw[line width=1.2] (a4)  --  (b4);
            \draw[line width=1.2] (a5)  --  (b5);
            \draw[line width=1.2] (a6)  --  (b6);
            \draw[line width=1.2] (a7)  -- (b17);
            \draw[line width=1.2] (a57) -- (b67);
            \draw[line width=1.2] (a27) --  (b7);

            \dashedline{1.2}{b17}{a57}{2.7}{1.8}{1}{1}
            \dashedline{1.2}{b67}{a27}{2.7}{1.8}{1}{1}

            \draw[line width=3.2, color=white] (c1a) -- (c1b);
            \draw[line width=3.2, color=white] (c2a) -- (c2b);
            \draw[line width=3.2, color=white] (c3a) -- (c3b);

            \dashedline{0.8}{c1a}{c1b}{3}{3}{0}{0}
            \dashedline{0.8}{c2a}{c2b}{3}{3}{0}{0}
            \dashedline{0.8}{c3a}{c3b}{3}{3}{0}{0}

            \contourlength{1.5pt}
            \contournumber{64}
            \coordinate[label=left:\( I_{i_2} \)]                (l5)  at  (a5);
            \coordinate[label=left:\( I_{i_1} \)]                (l6)  at  (a6);
            \coordinate[label=left:{\contour{white}{\( I_j \)}}] (l7)  at  (a7);
            \coordinate[label=left:\( C_1 \)]                    (lc1) at (c1a);
            \coordinate[label=left:\( C_2 \)]                    (lc2) at (c2a);
            \coordinate[label=left:\( C_3 \)]                    (lc3) at (c3a);

          \end{tikzpicture}

        \caption{%
          A sequence of intervals (higher is earlier in the sequence) where the interval~\( I_j \) has an exposed part that is not an interval (dashed part of~\( I_j \)).
          The covered area \( C_1 \cup C_2 \cup C_3 \) of the other intervals has the interval component~\( C_2 \) that splits the exposed part of~\( I_j \) into multiple interval components.
          The intervals \( I_{i_1} \) and \( I_{i_2} \) cover the interval~\( C_2 \).
          Moving them below~\( I_j \) results in their exposed parts being empty and \( I_j \) having one fewer interval component in its exposed part.
        }

        \label{fig:subadditive}
      \end{figure}

      \cref{lem:subadd-property} lets us find a polynomial bound for the number of exposed parts that we have to consider for a solution.

      \begin{lemma}
        \label{lem:interval-exposed-subset}
        Let \( \E_\I' \subseteq \E_\I \) be the set of all non-empty exposed parts of~\( \I \) that are intervals.
        Then~\( \abs{\E_\I'} \in \O{n^2} \) and we can compute~\( \E_\I' \) in \( \O{n^3} \) time.
      \end{lemma}
      \begin{proof}
        Any non-empty exposed part in~\( \E_\I' \) must start and end at some start or end points of some intervals in~\( \I \).
        Hence,~\( \E_\I' \) contains at most \( \O{n^2} \) non-empty exposed parts.

        Consider such a candidate interval~\( E = \interval{a, b} \).
        To check whether~\( E \) is an exposed part in some sequence, we keep track of the intervals~\( I \in \I \) which contain~\( E \) as a subinterval.
        Clearly, any such candidate~\( I = \interval{a', b'} \) must satisfy \( a' \le a \) and \( b' \ge b \).
        Additionally, it has to satisfy~\( a' = a \) or there has to exist a covered interval that covers~\( \interval{a', a} \) and ends at~\( a \).
        To check the latter condition, we determine the covered interval~\( C \) that starts the earliest and ends at~\( a \).
        This can be done in~\( \O{n} \) time assuming the covered intervals have been precomputed (in \( \O{n^3} \) time by \Cref{lem:preproc}).
        Now, any interval~\( I \in \I \) that starts before~\( C \) cannot have~\( E \) as its exposed part in any sequence.
        Hence, in~\( \O{n} \) time we can discard all intervals which cannot have an exposed part starting at~\( a \).
        For the remaining candidates, we analogously check the conditions for~\( b' \) and then discard all remaining intervals which cannot have an exposed part ending at~\( b \).
        That is, we can check if~\( E \) is an exposed part in some sequence in \( \O{n} \) time.
      \end{proof}

      This lets us formulate our first result.

      \begin{theorem}
        \label{thm:interval-exposed}
        Any instance of \fIO{} that has an optimal sequence where all exposed parts are intervals can be solved in \( \O{n^5} \) time.
      \end{theorem}
      \begin{proof}
        By \cref{lem:interval-exposed-subset}, we can compute the set~\( \E_\I' \) of non-empty exposed parts that are intervals in \( \O{n^3} \) time and its size is bounded in \( \O{n^2} \).
        We simply use the set~\( \E_\I' \) instead of~\( \E_\I \) in \cref{alg:compute-opt} to find an optimal sequence in \( \O{n^3 \cdot \abs{\E_\I'}} = \O{n^5} \) time.
      \end{proof}

      \cref{thm:interval-exposed} together with \cref{lem:subadd-property} yield the following corollary.
      \begin{corollary}
        \label{cor:subadditive}
        \fIO{} can be solved in \( \O{n^5} \) time for every function~\( f \) where~\( f - \f{0} \) is subadditive.
      \end{corollary}

      Note that polynomial time is unlikely if we only require~\( f \) to be subadditive, as we show in \Cref{ssec:nphard} that this case is \NPhard{}.
      In the next subsection, we use a similar approach for the case when~\( f - \f{0} \) is superadditive.

    \subsection{Cost Functions where~\( f - \f{0} \) is Superadditive}
      \label{ssec:superadditive}

      If~\( f - \f{0} \) is superadditive, then \( \f{a + b} + \f{0} \geq \f{a} + \f{b} \) holds for all \( a, b \in \nonneg \) (for example, every convex function).
      This case is of special interest because such cost functions are somewhat natural for applications (like the function~\( 2^x \) in the original motivation).
      Again, we observe a special structure of optimal sequences for this case.

      \begin{lemma}
        \label{lem:sup-add-property}
        Let~\( f \) be a cost function such that~\( f - \f{0} \) is superadditive.
        Then there is an optimal sequence where no interval comes before any of its proper subintervals.
      \end{lemma}
      \begin{proof}
        Assume that there is an optimal sequence where an interval~\( I_i \) comes after the interval~\( I_j \), but~\( I_i \) is a proper subset of~\( I_j \).
        Let~\( E_i = \es \) and~\( E_j \) be the exposed parts of~\( I_i \) and~\( I_j \) respectively.
        Consider what happens when we change the interval~\( I_i \) to be right before the interval~\( I_j \).
        Then, only the exposed parts of~\( I_i \) and~\( I_j \) can change.
        The exposed part of~\( I_i \) will be \( E_j \cap I_i \) and the exposed part of~\( I_j \) will be \( E_j \setminus I_i \) (see \cref{fig:superadditive}).
        By superadditivity, it holds:
        \begin{equation*}
          \label{eq:superadditive}
          \f{\abs{E_i}} + \f{\abs{E_j}} = \f{0} + \f{\abs{E_j \cap I_i} + \abs{E_j \setminus I_i}} \geq \f{\abs{E_j \cap I_i}} + \f{\abs{E_j \setminus I_i}}\text{.}
        \end{equation*}
        Therefore, the cost does not increase.
        We can repeat this step finitely many times until there are no more intervals before any of their proper subintervals.
        As we started with an optimal sequence and we did not increase the cost, this will be an optimal sequence as well.
      \end{proof}

      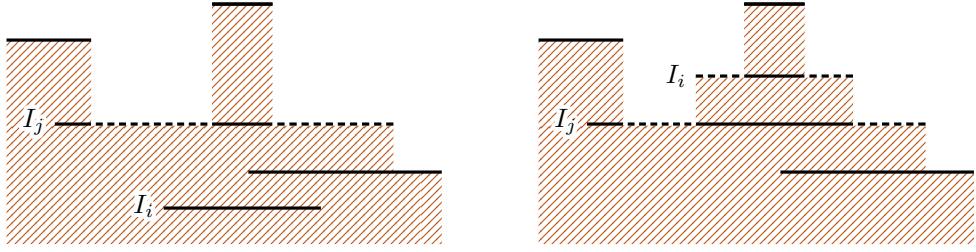
\begin{figure}[t]
        \centering
        \pgfmathsetmacro{\figscale}{\textwidth/22cm}
          \begin{tikzpicture}[scale=\figscale]

            \coordinate (a0)  at ( 0,    6   );
            \coordinate (a00) at ( 0,    0   );
            \coordinate (b0)  at (11,    6   );
            \coordinate (b00) at (11,    0   );
            \coordinate (a1)  at ( 5.25, 5   );
            \coordinate (a13) at ( 5.25, 2.5 );
            \coordinate (b1)  at ( 6.5,  5   );
            \coordinate (b13) at ( 6.5,  2.5 );
            \coordinate (a2)  at ( 1,    4.25);
            \coordinate (a20) at ( 1,    0   );
            \coordinate (b2)  at ( 2.75, 4.25);
            \coordinate (b23) at ( 2.75, 2.5 );
            \coordinate (a3)  at ( 2,    2.5 );
            \coordinate (b3)  at ( 9,    2.5 );
            \coordinate (b34) at ( 9,    1.5 );
            \coordinate (a4)  at ( 6,    1.5 );
            \coordinate (b4)  at (10,    1.5 );
            \coordinate (b40) at (10,    0   );
            \coordinate (a5)  at ( 4.25, 0.75);
            \coordinate (b5)  at ( 7.5,  0.75);

            \coordinate (shift) at (11,  0   );
            \coordinate (raise) at ( 0,  2.75);

            \coordinate (a0_)  at ($(a0)  + (shift)$);
            \coordinate (a00_) at ($(a00) + (shift)$);
            \coordinate (b0_)  at ($(b0)  + (shift)$);
            \coordinate (b00_) at ($(b00) + (shift)$);
            \coordinate (a1_)  at ($(a1)  + (shift)$);
            \coordinate (a15_) at ($(a13) + (shift) + (0, 1)$);
            \coordinate (b1_)  at ($(b1)  + (shift)$);
            \coordinate (b15_) at ($(b13) + (shift) + (0, 1)$);
            \coordinate (a2_)  at ($(a2)  + (shift)$);
            \coordinate (a20_) at ($(a20) + (shift)$);
            \coordinate (b2_)  at ($(b2)  + (shift)$);
            \coordinate (b23_) at ($(b23) + (shift)$);
            \coordinate (a3_)  at ($(a3)  + (shift)$);
            \coordinate (b3_)  at ($(b3)  + (shift)$);
            \coordinate (b34_) at ($(b34) + (shift)$);
            \coordinate (a4_)  at ($(a4)  + (shift)$);
            \coordinate (b4_)  at ($(b4)  + (shift)$);
            \coordinate (b40_) at ($(b40) + (shift)$);
            \coordinate (a5_)  at ($(a5)  + (shift) + (raise)$);
            \coordinate (a53_) at ($(a5)  + (shift) + (raise) - (0, 1)$);
            \coordinate (b5_)  at ($(b5)  + (shift) + (raise)$);
            \coordinate (b53_) at ($(b5)  + (shift) + (raise) - (0, 1)$);

            \fill[
                pattern=north east lines,
                pattern color=colorSecondary
            ]   (a20) -- (a2)%
                                 -- (b2)%
                        -- (b23)%
                        -- (a13) -- (a1)%
                                 -- (b1)%
                        -- (b13)%
                        -- (b3)%
                -- (b34)%
                -- (b4)%
            -- (b40);

            \fill[
                pattern=north east lines,
                pattern color=colorSecondary
            ]   (a20_) -- (a2_)%
                                  -- (b2_)%
                        -- (b23_)%
                        -- (a53_) -- (a5_)%
                                  -- (a15_) -- (a1_)%
                                            -- (b1_)%
                                  -- (b15_)%
                                  -- (b5_)%
                        -- (b53_)%
                        -- (b3_)%
                -- (b34_)%
                -- (b4_)%
            -- (b40_);

            \draw[line width=1.2] (a1)  --  (b1);
            \draw[line width=1.2] (a1)  --  (b1);
            \draw[line width=1.2] (a2)  --  (b2);
            \draw[line width=1.2] (a3)  -- (b23);
            \draw[line width=1.2] (a13) -- (b13);
            \draw[line width=1.2] (a4)  --  (b4);
            \draw[line width=1.2] (a5)  --  (b5);

            \draw[line width=1.2] (a1_)  --  (b1_);
            \draw[line width=1.2] (a1_)  --  (b1_);
            \draw[line width=1.2] (a2_)  --  (b2_);
            \draw[line width=1.2] (a3_)  -- (b23_);
            \draw[line width=1.2] (a53_) -- (b53_);
            \draw[line width=1.2] (a4_)  --  (b4_);
            \draw[line width=1.2] (a15_) -- (b15_);

            \dashedline{1.2}{b23}{a13}{2.7}{1.8}{1}{1}
            \dashedline{1.2}{b13}{b3}{2.7}{1.8}{1}{0}

            \dashedline{1.2}{b23_}{a53_}{2.7}{1.8}{1}{1}
            \dashedline{1.2}{b53_}{b3_}{2.7}{1.8}{1}{0}
            \dashedline{1.2}{a5_}{a15_}{2.7}{1.8}{0}{1}
            \dashedline{1.2}{b15_}{b5_}{2.7}{1.8}{1}{0}

            \contourlength{1.5pt}
            \contournumber{64}
            \coordinate[label=left:{\contour{white}{\( I_j \)}}] (l3) at (a3);
            \coordinate[label=left:{\contour{white}{\( I_j \)}}] (l3_) at (a3_);
            \coordinate[label=left:{\contour{white}{\( I_i \)}}] (l5) at (a5);
            \coordinate[label=left:\( I_i \)] (l5_) at (a5_);

          \end{tikzpicture}

        \caption{%
          Left: A sequence of intervals (higher is earlier in the sequence) where the interval~\( I_j \) comes before its subinterval~\( I_i \).
          The exposed part~\( E_j \) of the interval~\( I_j \) is drawn with dashed lines while the exposed part of~\( I_i \) is empty.
          Right: Moving \( I_i \) to be right before~\( I_j \) yields the exposed part \( E_j \setminus I_i \) for~\( I_j \) and \( E_j \cap I_i \) for \( I_i \).
        }
        \label{fig:superadditive}
      \end{figure}

      Next we show that the number of exposed parts to consider is polynomially bounded.

      \begin{lemma}
        \label{lem:sup-add-subset}
        Let~\( \E_\I' \subseteq \E_\I \) be the set of all non-empty exposed parts appearing in sequences where no interval comes before any of its proper subintervals.
        Then \( \abs{\E_\I'} \in \O{n^2} \) and we can compute~\( \E_\I' \) in \( \O{n^3} \) time.
      \end{lemma}
      \begin{proof}
        As regards the size of~\( \E_\I' \), consider an arbitrary exposed part~\( E \in \E_\I' \) of an interval~\( I \) and let~\( s \) be the start point of the first interval component of~\( E \) and let~\( t \) be the end point of the last interval component.
        Note that there cannot be another different exposed part also starting at~\( s \) and ending at~\( t \) since all proper subintervals of~\( \interval{s, t} \) in~\( \I \) must come before any interval with such exposed part.
        That is, each exposed part in~\( \E_\I' \) is uniquely defined by the start point of its first interval component and the end point of its last interval component.
        This yields~\( \abs{\E_\I'} \in \O{n^2} \).

        To compute~\( \E_\I' \), we iterate over all intervals~\( I = \interval{a, b} \in \I \).
        We first compute the covered area~\( C \) of all proper subintervals of~\( I \) in~\( \I \) in~\( \O{n} \) time (assuming the intervals are sorted) and add the exposed part~\( E \coloneqq I\setminus C \) to~\( \E_\I' \) in~\( \O{n} \) time.
        Then, we iterate over each pair of intervals~\( \interval{a', s} \in \I \)~and \( \interval{t, b'} \in \I \) with~\( a' < a \le s < t \le b < b' \) and remove them from~\( E \) and add the result to~\( \E_\I' \).
        As the exposed parts are uniquely defined by~\( s \) and~\( t \), we can check for duplicates in constant time when storing exposed parts via pairs~\( \r{s, t} \) in an~\( n\times n \) matrix.
        Overall, this takes~\( \O{n^3} \) time.
      \end{proof}
\noindent
      So, again we obtain a polynomial-time algorithm.

      \begin{theorem}
        \label{thm:superadditive}
        \fIO{} can be solved in \( \O{n^5} \) time for every~\( f \) where~\( f - \f{0} \) is superadditive.
      \end{theorem}
      \begin{proof}
        By \cref{lem:sup-add-subset}, we can compute the set~\( \E_\I' \) in \( \O{n^3} \) time.
        Iterating over~\( \E_\I' \) in \cref{alg:compute-opt}, we obtain a solution in \( \O{n^3\cdot\abs{\E_\I'}} = \O{n^5} \) time.
      \end{proof}

      Again, it is unlikely that polynomial time can be achieved for all superadditive functions, since this case is \NPhard{} (as shown in \Cref{ssec:nphard}).

    \subsection{Pairwise Intersecting Intervals}
      \label{ssec:pairwise-intersecting}

      So far, we developed polynomial-time algorithms for special cost functions.
      As a side result, we show that instances consisting of intervals that pairwise touch or intersect each other are also solvable in polynomial time for every function~\( f \).
      Note that in this case all covered areas are intervals.
      Using~\cref{alg:compute-opt}, one can solve such instances in polynomial time.

      \begin{theorem}
        \label{thm:pairwise-intersecting}
        Any \fIO{} instance where each covered area is a single interval can be solved in \( \O{n^6} \) time.
      \end{theorem}
      \begin{proof}
        As all covered areas are intervals, there are only \( \O{n^2} \) many.
        Hence, there are \( \O{n^2} \) exposed parts for every interval and we can find them by just going through all covered intervals.
        Additionally, all exposed parts have at most two interval components.
       Thus, \( \abs{\E_\I} \in \O{n^3} \) and \cref{alg:compute-opt} solves the instance in \( \O{n^3 \cdot \abs{\E_\I}} = \O{n^6} \) time.
      \end{proof}

      We remark that the running time can be improved to~\( \O{n^3} \) with a more specialized dynamic programming approach than our general \cref{alg:compute-opt}.
      As all covered areas are intervals, we can use the following decomposition:
      \[
        \opt{\I_{C}} = \min\c{\opt{\I_{C'}} + \f{\abs{I\setminus C'}} + \r{\abs{\I_C} - \abs{\I_{C'}} - 1}\cdot\f{0}}\text{,}
      \]
      where \( I \in \I_C \), \( C \ne C' \in \C_{\I_C} \) and \( C' \cup I = C \).
      (This is equivalent to guessing the next interval \( I \) in an optimal sequence that has covered \( C' \) so far.)
      Like in \cref{alg:compute-opt}, we have a table of size \( \O{n^2} \).
      But here we can iterate over \( \O{n} \) intervals instead of \( \O{n^3} \) exposed parts.
      Note, that \citetdurr{} used a similar approach to achieve \( \O{n^3} \) time for agreeable intervals.
      Our more general result in \cref{thm:interval-exposed} solves agreeable intervals in \( \O{n^5} \) time.
      Clearly, there is room to improve our polynomial running times.

  \section{Hardness Results}
    \label{sec:hardness}

    In this section, we complement our algorithmic results from the previous sections with running time lower bounds and \NPhard{ness}.

    \subsection{Exponential Running Time Lower Bound}
      \label{ssec:runtime-lb}
      We show that the exponential part of~\( \O{2^n} \) in the running time of~\cref{alg:compute-opt} from \cref{sec:arbitrary-cost} is basically optimal.
      More precisely, we prove that every deterministic algorithm that computes an optimal interval ordering for every~\( f \) without knowledge about~\( f \) (that is, it is only provided an \emph{oracle} for~\( f \) which outputs~\( \f{x} \) in constant time) requires at least~\( 2^{n-1} \) oracle queries (in the worst case).
      To prove this lower bound, we show that~\( f \) must be evaluated at all distinct lengths of exposed parts, of which there can be exponentially many.

      \begin{lemma}
        \label{lem:lower-bound}
        For all \( n \ge 2 \), there exists a set~\( \I \) of~\( n \) intervals with~\( \abs{\c{\abs{E} \mid E \in \E_\I}} \ge 2^{n-1} \) such that, for each~\( E \in \E_\I \), there is an ordering of~\( \I \) where no interval has an exposed part of length~\( \abs{E} \).
      \end{lemma}
      \begin{proof}
        Let \( \I = \c{I_1, \ldots, I_n} \) with \( I_i \coloneqq \interval{2^i, 2^{i+1}} \) for all \( i \in \s{n - 1} \) and \( I_n = \interval{0, 2^n} \).
        Since~\( I_1, \ldots, I_{n-1} \) are pairwise disjoint subintervals of~\( I_n \), each~\( I_i \) with~\( i \in \s{n - 1} \) has exposed parts of length~\( 2^i \) and~\( 0 \).
        Moreover,~\( I_n \) has exposed parts of length \( 2^n - \sum_{i \in X} 2^i \), where \( X \subseteq \c{1, \ldots, n - 1} \).
        Clearly, there are~\( 2^{n-1} \) possible sets~\( X \) and each \( \sum_{i \in X}2^i \) is unique (an even number between 0 and~\( 2^n - 2 \)).
        Hence, the possible lengths of exposed parts are all even numbers between 0 and~\( 2^n \).

        For each of these lengths, there is an interval ordering where no exposed part has this length.
        To see this, note that the sequence~\( (I_1, \ldots, I_n) \) has only exposed parts of length~\( 2^i \) for~\( i \in \s{n - 1} \).
        And for each~\( i \in \s{n - 1} \), the sequence~\( (I_1, \ldots, I_{i-1}, I_{i+1}, \ldots, I_n, I_i) \) has no exposed part of length~\( 2^i \).
      \end{proof}

      We now show that an algorithm with only oracle access to~\( f \) has to evaluate~\( f \) at all distinct lengths of exposed parts.

      \begin{theorem}
        \label{thm:lower-bound}
        Let~\( A \) be a deterministic algorithm computing an optimal interval ordering for arbitrary~\( f \) with oracle access to~\( f \) only.
        Then,~\( A \) requires at least~\( 2^{n-1} \) oracle queries in the worst case.
      \end{theorem}
      \begin{proof}
        Consider the interval set~\( \I \) from \cref{lem:lower-bound} for~\( n \ge 2 \) and assume towards a contradiction that~\( A \) outputs an optimal sequence~\( \sigma \) after less than~\( 2^{n-1} \) oracle queries, that is, without evaluating~\( f \) for all distinct lengths of exposed parts.
        Let~\( L \) be the length of an exposed part for which~\( \f{L} \) is never queried by~\( A \).
        We define a new cost function~\( f^* \) which only differs from~\( f \) at the value~\( L \) such that~\( \sigma \) is suboptimal for~\( f^* \) on~\( \I \).
        Note that~\( A \), however, will still output~\( \sigma \) on~\( \I \) for~\( f^* \) since it is deterministic and never queries the value~\( \fu{f^*}{L} \) and~\( f^* \) equals~\( f \) on all other values (hence, we obtain a contradiction).

        If an exposed part of an interval for~\( \sigma \) has length~\( L \), then we set \( \fu{f^*}{L} \coloneqq M > 0 \).
        Clearly, for large enough~\( M \), every sequence involving the exposed part length~\( L \) is not optimal for~\( f^* \) since there also exists a sequence avoiding all exposed parts of lengths~\( L \) (by construction of~\( \I \)).

        Otherwise, if no exposed part appearing in~\( \sigma \) has length~\( L \), then we set \( \fu{f^*}{L} \coloneqq M < 0 \) small enough.
        Clearly, now every sequence not involving an exposed part of length~\( L \) is suboptimal for~\( f^* \) (by analogous arguments as above).
      \end{proof}

      Note that the above argument relies on choosing an appropriate function~\( f^* \) which must have specific values at only finitely many points.
      Hence, the lower bound also holds for algorithms that solve \fIO{} (with oracle access) only for a certain class of functions such as piecewise constant or polynomials.

    \subsection{NP-Hardness}
      \label{ssec:nphard}

      \citetdurr{} showed \NPhard{ness} for the cost function~\( f \) where~\( \f{x} = 0 \) if~\( x \) is a power of two and~\( \f{x} = x \) otherwise.
      This function is rather artificial not being superadditive or monotone or continuous (however, it works for integral intervals).
      They reduced from the well-known \NPhard{} \nameref{pr:partition} problem~\cite{karp1972reducibility}.

      \begin{problem}[{Partition}]
        \label{pr:partition}
        \pinput{A multiset \( P = \c{p_1, \ldots, p_n} \) of positive integers.}
        \pquestion{Is there a subset \( P' \subseteq P \) such that \( \sum_{p \in P'}{p} = \tfrac{1}{2} \cdot \sum_{i=1}^{n}{p_i} \)?}
      \end{problem}\noindent

      We strengthen the \NPhard{ness} of \nameref{pr:f-interval-ordering}
      including continuous piecewise linear cost functions.

      \begin{theorem}
        \label{thm:general-reduction}
        \fIO is \NPhard{} for all functions \( f \) which satisfy the following: There exist \( \varepsilon > 0 \), \( x^* \ge 3\varepsilon \), \( c_1 < c_2 < c_3 \), and~\( h \),~\( h' \) such that
        \begin{itemize}
          \item \( \f{x} = c_2\cdot x \) for \( x \in \s{0, 2\varepsilon} \),
          \item \( \f{x} = c_1\cdot x + h \) for \( x \in \s{x^* - \varepsilon, x^*} \), and
          \item \( \f{x} = c_3\cdot x + h' \) for \( x \in \s{x^*, x^* + \varepsilon} \).
        \end{itemize}
      \end{theorem}
      \begin{proof}
        Consider an instance \( P = \c{p_1, \ldots, p_n} \) of \nameref{pr:partition}.
        We build an instance for \nameref{pr:f-interval-ordering}.
        Let \( S_k \coloneqq \sum_{i=1}^{k}{p_i} \) for all \( k \in \c{1, \ldots, n} \) and \( S_0 \coloneqq 0 \).
        We add the intervals \( I_k \coloneqq \interval{\tfrac{2\varepsilon S_{k-1}}{S_n}, \tfrac{2\varepsilon S_k}{S_n}} \) for all \( k \in \c{1, \ldots, n} \) to our instance.
        These intervals are pairwise disjoint and cover the interval \( \interval{0, 2\varepsilon} \).
        They represent the integers from \( P \), as \( \abs{I_k} = \tfrac{2\varepsilon  p_k}{S_n} \) holds for all \( k \in \c{1, \ldots, n} \).
        We further add the interval \( I_0 \coloneqq \interval{0, x^* + \varepsilon} \) and set~\( W \coloneqq \f{\varepsilon} + \f{x^*} \).
        Note that \( I_0 \) contains all other intervals.
        So, the exposed length of~\( I_0 \) will be \( \abs{I_0} - \sum_{I \in \I'}^{}{\abs{I}} \in \s{x^* - \varepsilon, x^* + \varepsilon} \), where~\( \I' \) is the set of intervals that come before~\( I_0 \).

        If~\( P \) has a solution~\( P' \subseteq P \), then we order the intervals such that all intervals~\( I_j \) with~\( p_j \in P' \) come before~\( I_0 \) and the remaining intervals come after~\( I_0 \).
        Then, the exposed part of~\( I_0 \) has length
        \[
          x^* + \varepsilon - \sum_{p_j \in P'}\abs{I_j} = x^* + \varepsilon - \sum_{p_j \in P'}\frac{2\varepsilon p_j}{S_n} = x^* + \varepsilon - \frac{2 \varepsilon}{S_n} \frac{S_n}{2} = x^*\text{.}
        \]
        Hence, the overall cost is
        \[
          \sum_{p_j \in P'}\f{\abs{I_j}} + \f{x^*} = c_2\sum_{p_j \in P'}\abs{I_j} + \f{x^*} = c_2\varepsilon + \f{x^*} = \f{\varepsilon} + \f{x^*} = W\text{.}
        \]

        Conversely, assume that there is an interval sequence~\( \sigma \) which has cost at most~\( W \).
        Then, let~\( I' \) be the subset of intervals coming before~\( I_0 \) in~\( \sigma \) and let~\( l_0 \) be the length of the exposed part of~\( I_0 \).
        If~\( l_0 = x^* + t \) for some~\( t \in (0, \varepsilon] \), then the cost for~\( I_0 \) is~\( \f{x^* + t} = \f{x^*} + c_3\cdot t \).
        Moreover, the intervals in~\( I' \) have total length~\( \varepsilon - t \) and their cost is~\( \f{\varepsilon - t} = c_2\cdot\r{\varepsilon - t} \).
        Thus, the overall cost is \( \f{x^*} + c_2\varepsilon + t\r{c_3 - c_2} = W + t\r{c_3-c_2} > W \) since~\( c_3 > c_2 \).
        Hence, this is not possible.
        Analogously, if~\( l_0 = x^* - t \) for \( t \in \left(0, \varepsilon\right] \), then the overall cost would be \( \f{x^*} - c_1\cdot t + c_2\cdot\r{\varepsilon + t} = \f{x^*} + c_2\varepsilon + t\r{c_2 - c_1} = W + t\r{c_2-c_1} > W \) (as \( c_2 > c_1 \)), which is again a contradiction.
        Hence, it follows that~\( l_0 = x^* \) which implies that the intervals in~\( I' \) have total length~\( \varepsilon \).
        Hence, they correspond to a solution for~\( P \).
      \end{proof}
      
      Note that in \Cref{thm:general-reduction} we only fixed the order of the slopes~\( c_1 \),~\( c_2 \), and~\( c_3 \), but not their actual values.
      As all these values may be positive (or negative) we obtain the following.
      \begin{corollary}
        \label{cor:general-reduction}
        \fIO is \NPhard{} for continuous strictly increasing (decreasing) piecewise linear cost functions~\( f \) (with only three pieces).
      \end{corollary}
      Moreover, \Cref{thm:general-reduction} also implies \NPhard{ness} for subadditive (and also superadditive) functions.
      To see this, observe that \Cref{thm:general-reduction} shows \NPhard{ness} for the function~\( \f{x} = \min\r{0, \abs{x - 2} - 1} \) (by choosing~\( \varepsilon = 0.5 \), \( x^* = 2 \), \( c_1 = -1 \), \( c_2 = 0 \), \( c_3 = 1 \), \( h = 1 \), \( h' = -3 \)).
      Note that~\( -1 \le \f{x} \le 0 \) holds for all~\( x \in \interval{0, \infty} \), and thus the function~\( f + 2 \) is subadditive since
      \[
        \f{x} + 2 + \f{y} + 2 \ge -1 + 2 - 1 + 2 = 2 \ge \f{x + y} + 2\text{.}
      \]
      In addition, the function~\( f - 1 \) is superadditive since
      \[
        \f{x} - 1 + \f{y} - 1 \le 0 - 1 + 0 - 1 = -2 \le \f{x + y} - 1\text{.}
      \]
      Clearly, adding a constant~\( c \in \real \) to the cost function does not change the structure of the problem (the optimal cost just changes by~\( n\cdot c \)).
      Hence, we obtain the following corollary.
      \begin{corollary}
        \label{cor:sub-superadditive-hardness}
        \fIO is \NPhard{} for subadditive (and superadditive) cost functions~\( f \).
      \end{corollary}

      Finally, we remark that our reductions do not work for integral input intervals (unlike the reduction by~\citetdurr{}).
      But we believe that our approach could be adapted to yield \NPhard{ness} for continuous monotonic piecewise linear functions with integral intervals.
      The cost function would need a series of such ``kinks'' (like the one at~\( x_0 \)) at increasingly large scales.

  \section{Parameterized Complexity}
    \label{sec:parameter}

    In this section, we introduce a more fine-grained analysis of the computational complexity of~\fIO{}.
    To this end, we identify certain parameters which render the problem tractable when their values are small (constants).
    As the exponential part of the running time of our general algorithm depends on the number~\( \abs{\E_\I} \) of possible exposed parts, we aim for parameters which bound this number.

    Our first parameter is the maximum number of interval components of any exposed part in~\( \E_\I \), which we call~\( \alpha_\I \).
    Note that we have already seen a polynomial-time algorithm for the case that \( \alpha_\I = 1 \) in \cref{thm:interval-exposed}.
    We generalize this result to any constant value of~\( \alpha_\I \) (in terms of parameterized complexity, the problem is in XP for parameter~\( \alpha_\I \)~\cite{downey2013fundamentals}).

    \begin{theorem}
      \label{thm:parameter-alpha}
      The size of~\( \E_\I \) is bounded in \( \O{n^{2\alpha_\I}} \).
      Thus, \fIO{} is polynomial-time solvable whenever~\( \alpha_\I \) is constant.
    \end{theorem}
    \begin{proof}
      Each interval component of an exposed part has at most \( n \) possible start and end points.
      Thus, each exposed part has at most \( n^2 \) possible interval components.
      So, for each non-empty exposed part, we have at most \( \alpha_\I \) choices, each being one of the at most~\( n^2 \) interval components.
      This yields \( \O{n^{2\alpha_\I}} \) different exposed parts.

      To compute~\( \E_\I \), we iterate over all intervals~\( I \) in~\( \I \) and compute all covered intervals that are contained in~\( I \).
      As any exposed part has at most~\( \alpha_\I \) interval components, we can partition the covered intervals into at most~\( \alpha_\I - 1 \) sets of pairwise intersecting or touching covered intervals.
      Then we just iterate over all combinations of covered intervals to find all exposed parts.
      To prevent duplicates, we keep the exposed parts in a balanced search tree.
      Thus, we can compute the set~\( \E_\I \) in~\( \O{n^{2\alpha_\I + c}} \) time for some constant~\( c \).
    \end{proof}

    It is unclear whether the exponential part of the running time can be solely confined to~\( \alpha_\I \), that is, whether a running time of~\( g(\alpha_\I)\cdot n^{\O{1}} \) is possible (this is called \emph{fixed-parameter tractability} in parameterized complexity~\cite{downey2013fundamentals}).
    However, for the parameter~\( s_\I \) defined as the maximum number of proper subintervals any interval~\( I \in \I \) has in~\( \I \), fixed-parameter tractability can be shown (note that agreeable intervals have~\( s_\I = 0 \)).

    \begin{theorem}
      \label{thm:fpt-subintervals}
      The size of~\( \E_\I \) is bounded in \( \O{n^2 \cdot 2^{s_\I}} \).
      Thus,~\fIO{} is fixed-parameter tractable with respect to~\( s_\I \).
    \end{theorem}
    \begin{proof}
      Clearly, there are \( \O{n^2} \) possible pairs of first start and last end points of the interval components of any exposed part.
      The interval spanned by the first start and last end point of any exposed part has at most~\( s_\I \) subintervals in~\( \I \).
      Hence, there are at most \( 2^{s_\I} \) different exposed parts for each combination of first start and last end point.
      So, overall there are \( \O{n^2 \cdot 2^{s_\I}} \) possible exposed parts.
      It is easy to see that we can compute the set~\( \E_\I \) in \( \O{n^c \cdot 2^{s_\I}} \) time for some constant~\( c \), as we just have to try all combinations and, like in \cref{lem:compute-exposed}, keep all exposed parts we find in a balanced search tree to detect duplicates in polynomial time.
    \end{proof}

  \section{Conclusion}
    \label{sec:conclusion}

    We took a large step towards settling the computational complexity of the interval ordering problem.
    As regards the cost function, we showed broad families of functions to be polynomial-time solvable while we also showed \NPhard{ness} for a wide range of functions.
    Further closing the gap between the tractable and intractable cases by identifying the specific properties that render a function tractable is an interesting challenge.
    For example, the complexity for subadditive and superadditive functions is not fully settled.

    Regarding running times, on one side, there is clearly the challenge to improve the polynomial running times.
    This may allow experiments on real-world datasets in order to evaluate whether integrating an \fIO{} solver into actual protein backbone reconstruction algorithms is worth the computation overhead.
    On the other side, proving conditional running time lower bounds for the polynomial-time solvable cases would also be interesting.
    Moreover, we initiated the study of the parameterized complexity of the problem.
    It could be fruitful to identify parameters which render instances tractable in practice.
    For example, it is open whether the problem is fixed-parameter tractable for the maximum number~\( \alpha_\I \) of interval components of any exposed part.
    



  \bibliography{references}


\end{document}